\documentclass[11pt]{article}
\usepackage[totalwidth=13.0cm,totalheight=20.0cm]{geometry}
\usepackage{latexsym,amsthm,amsmath,amssymb,url}
\usepackage[ruled, linesnumbered]{algorithm2e}
\usepackage{enumerate}
\usepackage{tikz}

\usetikzlibrary{decorations.pathreplacing}

\newtheorem{lemma}{Lemma}

\newtheorem{definition}{Definition}
\newtheorem{theorem}{Theorem}

\usepackage{authblk}
\title{FPT algorithms for generalized feedback vertex set problems}
\author{Bin Sheng \thanks{Email: shengbinhello@nuaa.edu.cn}}
\affil{College of Computer Science and Technology, Nanjing University of Aeronautics and Astronautics, Collaborative Innovation Center of Novel Software Technology and Industrialization, Nanjing, Jiangsu, 211106, PR China}
\begin{document}

\maketitle

\begin{abstract}
%The graph class, $r$-pseudoforests, in which each component can be made into a forest by deleting at most $k$ edges, was raised by Philip et al. [MFCS 2015], generalizing the concept of forests. They study the problem of deleting minimum number of vertices to obtain an $r$-pseudoforest from the FPT perspective, which is a generalization of the well studied feedback vertex set problem. Philip et al. show that the problem is FPT parameterized with $k$, the number of vertices allowed to be deleted. In this paper, we provide improved FPT algorithm and kernelization results for the $r$-pseudoforest deletion problem.   We also give an FPT algorithm for the problem of $d$-quasi-forest deletion, which asks to delete minimum number of vertices to obtain a $d$-quasi-forest, in which each component admits a feedback vertex set of size at most $d$.

An $r$-pseudoforest is a graph in which each component can be made into a forest by deleting at most $r$ edges, and a $d$-quasi-forest is a graph in which each component can be made into a forest by deleting at most $d$ vertices.

In this paper, we study the parameterized tractability of deleting minimum number of vertices to obtain $r$-pseudoforest and $d$-quasi-forest, generalizing the well studied feedback vertex set problem. We first provide improved FPT algorithm and kernelization results for the $r$-pseudoforest deletion problem and then we show that the $d$-quasi-forest deletion problem is also FPT.

%\textcolor{red}{We can also try the two problems parameterized with number of high-degree-vertices.}
%\textcolor{red}{Maybe we can give a randomized algorithm by arguing that at least constant fraction of the edges are incident with the solution.}
\end{abstract}

\section{Preliminary}
The Feedback Vertex Set problem (FVS for short), which asks to delete minimum number of vertices from a given graph to make it acyclic, is one of the 21 NP-hard problems proved by Karp \cite{DBLPconf/coco}. It has important applications in  bio-computing, operating system and artificial intelligence and so on. The problem has attracted a lot of attention from the parameterized complexity community due to its importance. Both its undirected and directed version have been well studied \cite{Cao2015On,DBLP:journals/jcss/ChenFLLV08,DBLP:journals/talg/ChitnisCHM15,DBLP:journals/siamdm/CyganPPW13}.

Feedback vertex set problem has been proved to be fixed parameter tractable when parameterized with the solution size, that is, the number of vertices allowed to be deleted. For undirected feedback vertex set problem, the state of the art algorithm runs in time $O^*(3.460^k)$ in deterministic setting \cite{DBLP:conf/iwpec/IwataK19} and $O^*(3^k)$ in the randomized setting \cite{DBLP:journals/corr/abs-1906-12298}, here the $O^*$  notation hides polynomial factors in $n$.

Several classes of graphs that are nearly acyclic have been defined in the litarature. A graph $F$ is called an \textit{$r$-pseudoforest} if we can delete at most $r$ edges from each component in $F$ to get a forest. A \textit{pseudoforest} is a 1-pseudoforest.  A graph $F$ is called an \textit{almost $r$-forest} if we can delete $r$ edges from $F$ to get a forest.

As a generalization of feedback vertex set problem, Philip et al. \cite{DBLP:journals/siamdm/PhilipRS18} introduced the problem of deleting vertices to get a graph that is nearly a forest. Several results have been obtained in this line of research. In \cite{DBLP:journals/siamdm/PhilipRS18}, the authors gave a $O(c_{r}^{k}n^{O(1)})$ algorithm for $r$-pseudoforest deletion, which asks to delete at most $k$ vertices to get an $r$-pseudoforest. The $c_{r}$ here depends on $r$ doubly exponentially. They also gave a $7.56^kn^{O(1)}$ time algorithm for the problem of pseudoforest deletion.

Rai and Saraub \cite{DBLP:journals/tcs/RaiS18} gave a $O^{*}(5.0024^{(k+r)})$ algorithm for the Almost Forest Deletion problem, which asks to delete minimum vertices to get an almost $r$-forest. Lin et al. \cite{DBLP:journals/ipl/LinFWCFL18} gave an improved algorithm for this problem that runs in time $O^{*}(5^k4^r)$.
Bodlaender et al. \cite{DBLP:journals/dam/BodlaenderOO18} gave an improved algorithm for pseudoforest deletion running in time $O(3^knk^{O(1)})$.

%\begin{theorem} \cite{DBLP:journals/siamdm/PhilipRS18}\label{theorem1}$r$-{\sc pseudoforest deletion} can be solved in time $O(c_{r}^k(m+n)\log n)$ for an instance $(G,k)$ where the constant $c_r$ depends only on $r$.\end{theorem}

A $d$-quasi-forest is a graph in which each connected component admits a feedback vertex set of size at most $d$. This notion was raised by  Hols and Kratsch in \cite{DBLP:conf/iwpec/HolsK17} in which they show that the Vertex Cover problem admits a polynomial kernel when parameterized with distance to $d$-quasi-forest. They did not show how to obtain such a modulator to $d$-quasi-forest.

In this paper, we first give an algorithm for $r$-pseudoforest deletion that runs in time $(1+(2r+3)^{r+2})^{k+1}n^{O(1)}$, improving  the algorithm in \cite{DBLP:journals/siamdm/PhilipRS18}. We also give an FPT algorithm to obtain a minimum modulator to $d$-quasi-forest. To the author's acknowledgement, this is the first nontrivial FPT result for $d$-quasi-forest deletion.

\section{Notations and Terminology}
Here we give a brief list of the graph theory concepts used in this paper, for other notations and terminology, we refer readers to \cite{bollobas2013modern}.

For a graph $G=(V(G),E(G))$, $V(G)$ and $E(G)$ are called its vertex set and edge set respectively.
A non-empty graph $G$ is \textit{connected} if there is a path between any pair of vertices, otherwise, we call it \textit{disconnected}.

The \textit{multiplicity} of an edge is the number of copies it appears in the multigraph. An edge $uv$ is called a \textit{loop} if $u=v$. The \textit{degree} of a vertex is the number of edges incident with it. A \textit{forest} is a graph in which there is no cycle. A \textit{tree} is a connected forest.  A graph $H=(V(H),E(H))$ is a \textit{subgraph} of a graph $G=(V(G),E(G))$, if $V(H)\subseteq V(G)$ and $E(H)\subseteq E(G)$. A subgraph $H$ of $G$ is called an \textit{induced subgraph} of $G$ if for any $u,v\in V(G)$, edge $uv\in E(H)$ if and only if $uv\in E(G)$. We also denote $H$ by $G[V(H)]$ to indicate that $H$ is induced by the vertex set $V(H)$.

A \textit{pseudoforest} is a graph in which each connected component can be made into a forest by deleting at most one edge. For a positive integer $r$, an \textit{r-pseudoforest} is a graph in which each connected component can be made into a forest by deleting at most $r$ edges.
A \textit{$d$-quasi-forest} is a graph in which each component admits a feedback vertex set of size at most $d$.

\section{Branching algorithm for $r$-pseudoforest deletion }
\begin{definition}
Given a graph $G=(V,E)$, a subset $S\subseteq V(G)$ is called an \textbf{$r$-pseudoforest deletion set} of $G$ if $G-S$ is an $r$-pseudoforest.
\end{definition}

Here is an iterative version of the parameterized $r$-pseudoforest deletion problem.

\textbf{$r$-pseudoforest Deletion}

    \emph{Instance:} Graph $G$ with an $r$-pseudoforest deletion set $S$, $|S|\leq k + 1$, integers $k$ and $r$.

    \emph{Parameter:} $k$ and $r$.

    \emph{Output:} Decide if there exists an $r$-pseudoforest deletion set $X$ of $G$ with $|X|\leq k$?

For a connected component $C$ in a graph $G$, we call the
quantity $|E(C)| - |V(C)| + 1$ the \textit{excess} of $C$ and denote it by $ex(C)$. Note that $ex(C)\geq 0$ for any connected component. Let $\cal C$ be the set of components in $G$. We define the excess
of $G$, denoted by $ex(G)$, as $ex(G) = max_{C\in \cal C} ex(C)$£¬ i.e. the maximum excess among all components in $G$.
By definition, $G$ is an $r$-pseudoforest if and only if $ex(G) \leq r$.
For vertex subset $S\subseteq V(G)$, let $cc(S) = cc(G[S])$ be the number of components in $G[S]$.

We give the following observations about $r$-pseudoforest.

\textbf{Observation 1:} If $G'$ is a subgraph of an $r$-pseudoforest, then $G'$ is also an $r$-pseudoforest.

\textbf{Observation 2:} If $G$ is an $r$-pseudoforest, then each component $C$ in $G$ has at most $|V(C)| - 1 + r$ edges.

By Observation 2, checking whether a given graph is an $r$-pseudoforest can be done in polynomial time.

Now we show how to solve the $r$-pseudoforest deletion via the approach of iterative compression.
As a standard step, we introduce the following disjoint version of $r$-pseudoforest deletion.

\quad

\textbf{Disjoint $r$-pseudoforest Deletion}

Input: Graph $G$ with an $r$-pseudoforest deletion set $S$, $|S|\leq k+1$, integers $k$ and $r$.

Parameter: $k$ and $r$.

Output: Decide if there exists an $r$-pseudoforest deletion set $X$ of $G$ with $|X|\leq k$ and $X\cap S =\emptyset$?

\quad

To solve Disjoint $r$-pseudoforest Deletion, we apply the following reduction rules.

\textbf{Reduction Rule 1}: Let $(G,S,k,r)$ be an instance of Disjoint $r$-pseudoforest Deletion, if there exists a vertex $v\in V(G)\setminus S$ such that $d_{G}(v) = 1$, then return $(G-v,S,k,r)$.

%The correctness of Reduction Rule 1 follows from Lemma \ref{lamma:irrelevant}.
The correctness of Reduction Rule 1 is easy to prove.
\begin{proof}
We show $(G,S,k,r)$ and $(G-v,S,k,r)$ are equivalent instances of Disjoint $r$-pseudoforest Deletion by proving that they have same solutions.

On the one hand, if $X\subseteq V(G)\setminus S$ is an $r$-pseudoforest deletion set of $G$ disjoint from $S$, then $X$ is also an $r$-pseudoforest deletion set of $G-v$.  Indeed, if $v\in X$, then $G-v-X = G-X$, otherwise, $G-v-X$ is a subgraph of $G-X$, which is an $r$-pseudoforest graph. In both cases, $G-X$ is an $r$-pseudoforest graph.

On the other hand, let $X\subseteq V(G)\setminus S$ be an $r$-pseudoforest deletion set of $G-v$ disjoint from $S$. Let $C$ be the component in $G-X$ containing $v$. By definition of $X$, $C-v$ is an $r$-pseudoforest in $G-v-X$. Since $d_{G}(v) = 1$, it follows that $C$ is also an $r$-pseudoforest in $G-X$. Furthermore, $G-v-X$ and $G-X$ only differ at component $C$. Thus, $X\subseteq V(G)\setminus S$ is an $r$-pseudoforest deletion set of $G$ disjoint from $S$.

\end{proof}

\textbf{Reduction Rule 2}: If there exists $v \in V(G)\setminus S$ such that $G[S \cup {v}]$ is not an
$r$-pseudoforest, then return $(G-v,S,k-1,r)$.

It is obvious that any $r$-pseudoforest deletion set disjoint from $S$ must contain $v$.

\textbf{Reduction Rule 3}: If there exists a vertex $u \in V (G) \setminus S$ of degree two, such
that at least one neighbor of $u'$ is in $V (G) \setminus S$, then delete $u$ and put a new
edge between its two neighbors (even if they were already adjacent). If both incident edges of $u$
are to the same vertex, delete $u$ and put a new loop on the adjacent
vertex (even if it has loop(s) already).

Now we prove the correctness of Rule 3.
\begin{lemma}
Reduction Rule 3 is correct.
\end{lemma}
\begin{proof}
Note that the operation decreases both the edge number and vertex number by 1.

On the one hand, let $(G,S,k,r)$ be a yes-instance of Disjoint $r$-pseudoforest Deletion, and let $X$ be an minimal $r$-pseudoforest deletion set of size at most $k$, which is disjoint from $S$. Let $u$ be the vertex of degree 2 being deleted from $G$ and $G'$ be the resulting graph. Let $x,y$ be the two neighbors of $u$ in $G$.

If $u\in X$, then by the minimality of $X$, the neighbor of $u$ not in $S$ is not in $X$. Without loss of generality, we may assume that $s\not\in S$ and $x\not\in X$, then $X'=X\setminus \{u\}\cup \{x\}$ is an  $r$-pseudoforest deletion set of $G'$. Thus if $(G,S,k,r)$ is a yes instance, then $(G',S,k,r)$ is also a yes instance. %\textcolor{red}{need to be careful here, as we need $x\not\in S$}

If $u\not\in X$, then $u$ is an isolated vertex or in the component containing $x$ or $y$ in $G-X$. $G'-X$ can be obtained from $G-X$ by deleting $u$ or bypassing it, either way, $G'-X$ is an $r$-pseudoforest.

On the other hand, let $(G',S,k,r)$ be a yes instance of Disjoint $r$-pseudoforest Deletion, and let $X'$ be an minimal $r$-pseudoforest deletion set of size at most $k$, which is disjoint from $S$. If both $x,y$ are not in $X'$, then $G-X'$ can be obtained from $G'-X'$ by subdividing edge $xy$ and name the new vertex $u$. If at least one of $x$ and $y$ is in $X'$, then $u$ has degree 0 or 1 in $G-X$. In both cases, $G-X'$ is also an $r$-pseudoforest. It follows that $X'$ is also a solution for $(G,S,k,r)$.
\end{proof}

\textbf{Reduction Rule 4: If $k<0$, then return no.}

It is easy to see that Rules 1-4 can
be applied in polynomial time.
Given an instance $(G, S, k, r)$ of Disjoint $r$-pseudoforest Deletion, we first apply Rules 1-4 whenever possible.

Now we show how to handle the case when none of Rules 1-4 can be applied.

Define measure $\phi(I) = k+ cc(S) + \Sigma_{C \in {\cal C}(G[S])}(r - ex(C)).$ Note that initially
$\phi(I) \leq k + cc(S) + cc(S)r\leq 2k+(k+1)r+1< (k+1)(r+2)$ as $|S|\leq k+1$.

{In order to get a depth bounded search tree, we prove that $\phi(I)$ decreases after each application of the following branching rules.}

\quad

\textbf{BR-1. Branching on a vertex $v\notin S$ with $d_{S}(v) \geq 2$.}

In one branch, we put $v$ into the solution and call the algorithm on $(G-\{v\}, S, k-1, r)$. Note that in this branch, $cc(S)$, $ex(C)$ (for each $C \in {\cal C}(G[S])$) remain the same while $k$ decreases by 1. Hence $\phi(I)$ drops by 1.

In the other branch, we put $v$ into $S$ and call the algorithm on $(G, S \cup \{v\}, k, r)$. Let $S' = S\cup \{v\}$. There are the following two possible cases regarding the distribution of $N_{S}(v)$.

Case 1: $N_{S}(v)$ belongs to more than one component in $G[S]$, thus $cc(S')\leq cc(S) -1$. Let $C_1, C_2 ,\ldots, C_t (t\geq 2)$ be the set of components in $G[S]$ that are adjacent to $v$.

To compute the difference between excess sums in $S$ and $S'$, denote
\begin{eqnarray*}
\sigma_{S}  &=& \Sigma_{i \in [t]}(r - ex(C_i))\\
            &=& rt - \Sigma_{i \in [t]}(ex(C_i)) \\
            &=& rt - \Sigma_{i\in [t]} (|E(C_i)|-|V(C_i)|+1)\\
            &=& rt - \Sigma_{i\in [t]} |E(C_i)| + \Sigma_{i\in [t]} |V(C_i)|-t,
\end{eqnarray*}
\begin{eqnarray*}
\sigma_{S'}&=& r - ex(G[\cup_{i\in[t]}V(C_{i}) \cup \{v\}])\\
           &=& r - (|E(\cup_{i\in[r]}V(C_{i})\cup \{v\})| - |V(\cup_{i\in[r]}V(C_{i})\cup \{v\})| + 1) \\
           &=& r - (\Sigma_{i\in [r]} |E(C_i)| + d_{S}(v) - \Sigma_{i\in [r]}|V(C_i)|) %|V(C_i)|)$
\end{eqnarray*}
It follows that $\sigma_{S}-\sigma_{S'} = r(t-1)+ d_{S}(v) -t \geq r(t-1) \geq r$.
In this case, $\phi(I)$ drops by at least $1 + \sigma_{S}-\sigma_{S'} \geq 1+  r(t-1)\geq 1+r$.

Case 2: All the neighbors of $v$ belong to one component, denoted by $C^{*}$. Then $cc(S') = cc(S)$, and $ex(G[V(C^{*})\cup \{v\}])-ex(C^{*}) \geq 1$ as $d_{S}(v) \geq 2$. Hence $\phi(I)$ decrease by at least 1. %We still have that $ex(G[S'])$ does not exceed $r$, otherwise, we would have applied Reduction Rule 2.

Therefore, in BR-1, the measure $\phi(I)$ drops by $1$ in one case, and at least $1+ r$ or $1$ in the other, while remaining non-negative. In the worst case, it gives us branching vector $(1, 1)$.

\quad

Observe that after exhaustive applications of Rule 3 and BR-1, every vertex in $G-S$ has degree at least 3. Moreover, if there exists a vertex $u\not\in S$ such that $d_{G - S}(u)\leq 1$, then $d_{S}(u)\geq 2$. And so if $d_{S}(u)\leq 1$ holds for each $u\in G-S$, then $d_{G - S}(u)\geq 2$, that is $\delta(G-S)\geq 2$. Thus each $u\in G-S$ must be in a cycle.

If $G$ is not an $r$-pseudoforest, there must be edges between $G-S$ and $S$, as we know both $G[S]$ and $G-S$ are $r$-pseudoforest. In the following, we branch on vertices in $G-S$ adjacent to $S$.

\quad

\textbf{BR-2. Branching on a vertex $v \notin S$ adjacent to $S$.}

 First let us consider the case when there is a component $C$ in $G-S$ such that there is only one edge $uv$ between $C$ and $S$, where $u\in C$ and $v\in S$. {Note that if the solution should intersect $V(C)$, then it suffices to contain $u$.} Thus we may branch on whether $u$ is in the solution or not. If $u$ is in the solution, then $k$ decreases and thus $\phi(I)$ decreases by one. If $u$ is not in the solution, then we put $C$ into $S$, which greatly decreases the value of $r-ex(S)$, thus $\phi(I)$ also decreases by at least one.

 %If $S\cup C$ is not $r$-pseudoforest, then delete $u$ (and further we can ignore $C-u$) and decrease $k$ by 1. This is safe because we must delete at least one vertex from $G[C\cup S]$, and deleting $u$ suffices. If $S\cup C$ is an $r$-pseudoforest, branch on whether $u$ is deleted.

Now assume each component in $G-S$ has at least two edges to $S$.
Look at one shortest path $P$ in $G-S$, such that both endvertices of $P$ are adjacent to some vertex in $S$ (we allow $P$ to be an isolated vertex). Note that such a shortest path can be found in polynomial time. We prove that $|V(P)| \leq 2r+2$. As any component $C$ in $G-S$ is an $r$-pseudoforest, $|E(C)| - |V(C)| +1\leq r$. Note that each vertex in $P$ has degree at least 3 after exhaustive applications of Rule 3. And observe that no internal vertex of $P$ has an edge to $S$, otherwise we find a path shorter than $P$, a contradiction. Let $C_0$ be the component in $G-S$ containing $P$, we know  $|E(C_0)| \geq 3/2|V(P)|-2$, and $ex(C_0) \geq ex(G[V(P)]) \geq 3/2|V(P)|-2 - |V(P)| + 1$. As $ex(C_0) \leq r$, so $|V(P)| \leq 2r+2$.

We branch on whether to delete any vertex on the path $P$. Suppose $P= v_1,v_2\ldots, v_t$. We consider $t+1$ branches. In branch $i$ where $i\in [t]$, we delete vertex $v_i$ on path $P$, and call the algorithm on $(G-\{v_{i}\}, S, k-1, r)$. In branch $t+1$, we don't delete any vertex on $P$.  Note that for branch $i$, where $i\leq t$, $cc(S)$, $ex(C)$ (for any $C \in {\cal C}(G[S])$) remain the same while $k$ decrease by at least 1. Thus $\phi(I)$ drops by at least 1.

If $G[S+V(P)]$ is not an $r$-pseudoforest, then we ignore branch $t+1$.
Otherwise, for branch $t+1$, we get a new instance $(G, S', k, r)$, where $S' = S\cup V(P)$. If edges between $V(P)$ and $S$ are to the same component $C$ in $G[S]$, then $ex(C \cup V(P)) = ex(C) + 1$, thus in this branch, $\phi(I)$ decrease by 1. Otherwise, the edges between $V(P)$ and $S$ are to different components in $G[S]$.  In this case, $cc(S)$ decreases by 1 and $\sigma_{S}-\sigma_{S'}\geq 0$, so $\phi(I)$ drops by at least 1.
This gives us the $(t+1)$-tuple branching vector $(1, 1,\ldots,1)$ in which $t\leq 2r+2$.

According to the branching vectors in BR-1 and BR-2, the algorithm runs in time $O^{*}((2r+3)^{(k+1)(r+2)})$.

The following lemma states that a fast parameterized algorithm for the disjoint version problem gives a fast algorithm for the original problem.

\begin{lemma}\cite{DBLP:books/sp/CyganFKLMPPS15}
If there is an algorithm sloving Disjiont $r$-pseudoforest Deletion in time $f(k)n^{O(1)}$, then there is an algorithm solving $r$-pweudoforest Deletion in time $\sum_{i=0}^{i=k}\binom{k+1}{i}f(k-i)n^{O(1)}$.
\end{lemma}

 So we get an algorithm for $r$-pseudoforest deletion with running time
%$\sum_{i=0}^{i=k}\binom{k+1}{i}(2r+3)^{(k-i+1)(r+2)}< (2r+3)^{(2r+5)k}$,
$\sum_{i=0}^{i=k}\binom{k+1}{i}(2r+3)^{(k-i+1)(r+2)}< (1+(2r+3)^{r+2})^{k+1}$, which improves over the result in \cite{DBLP:journals/siamdm/PhilipRS18}.
Our result answers the question raised in \cite{DBLP:journals/dam/BodlaenderOO18} on whether there is an algorithm for $r$-pseudoforest deletion runs in time $O^{*}(c_{r}^{k})$.

\begin{theorem}
There exists an algorithm for $r$-pseudoforest deletion with running time %$(2r+3)^{(2r+5)k}n^{O(1)}$
$(1+(2r+3)^{r+2})^{k+1}n^{O(1)}$.
\end{theorem}

\section{Kernelization of $r$-pseudoforest deletion}
In this section, we give an improved kernel for the $r$-pseudoforest deletion problem. By exhaustively applying Reduction Rules 1 and 3, we will get an instance with minimum degree at least 3.

\begin{lemma}\label{lemma1}
If a graph $G$ has minimum degree at least 3, maximum degree at most $d$, and an $r$-pseudoforest deletion set of size at most $k$, then it has at most $(2dr-d+1)k$  vertices and at most $3kdr-kd$  edges.
\end{lemma}
\begin{proof}
Let $X$ be an $r$-pseudoforest deletion set of $G$ of size at most $k$. Let $F = G-X$. It follows that each component in $F$ can be made into a forest by deleting at most $r$ edges.
Suppose there are $c$ components in $F$, we know that $c \leq kd$, since deleting any vertex of degree $t$ produces at most $t$ components. For each
component $C_i$ with $i\in[c]$, we know there are at most $|V(C_i)|-1+r$ edges. Thus,
$|E(F)|=\Sigma_{i \in [c]}|E(C_i)| \leq \Sigma_{i \in [c]}(|V(C_i)|-1+r) = |V(F)|+c(r-1)$. By counting the number of edges incident with $V(F)$, we know that
$$
3|V(F)| \leq 2(|E(F)|)+|E(X,V(F))| \leq 2(|V(F)|+c(r-1))+kd.
$$
It follows that $|V(F)| \leq 2c(r-1)+kd$. So $|V(G)| \leq |X|+|V(F)| \leq 2c(r-1)+ k(d + 1)\leq (2dr-d+1)k$. And $|E(G)|\leq |E(F)|+|E(X,V(F))| + |E(G[X])| \leq |V(F)|+c(r-1)+kd < 3c(r-1)+2kd=3kdr-kd$.

 \end{proof}

 We need to make use of the following result.
\begin{theorem}[\cite{DBLP:journals/siamdm/PhilipRS18}]\label{theorem1}
 Given an instance $(G, k)$ of $r$-pseudoforest Deletion, in polynomial time, we can get an equivalent instance $(G', k')$ such that $k' \leq k$, $|V (G')| \leq |V (G)|$ and the maximum degree of $G'$ is at most $(k + r)(3r + 8)$.
\end{theorem}

\begin{theorem}\label{theorem2}
The $r$-pseudoforest deletion problem admits a kernel with $O(k^2r^2)$ vertices and $O(k^2r^2)$ edges.
\end{theorem}
\begin{proof}
According to Theorem \ref{theorem1}, we know that $r$-pseudoforest deletion admits a kernel with maximum degree at most  $d=(k + r)(3r + 8)$. Thus by Lemma \ref{lemma1}, we can obtain a kernel for $r$-pseudoforest deletion which has at most $(2r-1)kd+k=O(k^2r^2)$ vertices, and at most $(3r-1)kd = (3r-1)k(k + r)(3r + 8)=O(k^2r^2)$ edges.
\end{proof}
The kernel in Theorem \ref{theorem2} improves over the kernel in \cite{DBLP:journals/siamdm/PhilipRS18}, in which the kernel is of size $O(ck^2)$, where the constant $c$ depends on $r$ exponentially.

\section{$d$-quasi-forest deletion}

\begin{definition}
Given graph $G=(V,E)$, a subset $S\subseteq V(G)$ is called a \textbf{$d$-quasi-forest deletion set} of $G$ if $G-S$ is a $d$-quasi-forest.
\end{definition}

Let us recall the definition of $d$-quasi-forest deletion first.

 \quad

  { \textbf{$d$-quasi-forest deletion}} \nopagebreak

    \emph{Instance:} An undirected graph $G$, an integer $k$.

    \emph{Parameter:} $k$.

    \emph{Output:} Decide if there exists a set $X\subseteq V(G)$ with $|X|\leq k$ such that $G-X$ is a $d$-quasi-forest.

\quad

\begin{lemma}
A yes instance of $d$-quasi-forest deletion has treewidth at most $k+d+1$.
\end{lemma}
\begin{proof}
Let $G$ be a yes instance of $d$-quasi-forest deletion, then there is a set $X\subseteq V(G)$, such that each component in $G-X$ has feedback vertex set of size at most $d$. Thus each component in $G-X$ has treewidth at most $d+1$. By putting $X$ into each bag of the tree decomposition for each component in $G-X$, we obtain a tree decomposition of $G$, which has treewdith at most $k+d+1$. It follows that $G$ has treewidth at most $k+d+1$.
\end{proof}
{We first point out that that the problem is  FPT according to Courcelle's  theorem by expressing it with Monadic Second Logic.}
The basic idea of the expression is as follows: $\exists v_1,v_2,\ldots, v_k\in V(G)$ such that  $\forall X\subseteq V(G)-\{v_1,v_2,\ldots, v_k\}$, $Conn(X)\rightarrow FVS(X)\leq d$. The definition of $FVS(X)\leq d$ is as follows:
$\exists y_1,y_2,\ldots, y_d\in X$, such that $\neg$ $ExistsCycle(X-\{y_1,y_2,\ldots, y_d\})$. The definition of $ExistsCycle(X)$ is as follows: $\exists E\subseteq E(G)$, such that $Conn(E)$ and $\forall e\in E$, $\exists u,v\in X$, such that $Inc(u,e)$ and $Inc(v,e)$ and $Deg(u,E)=2$.

%\textcolor{blue}{To be sure of the correctness, we may write $\exists v_1,v_2,\ldots, v_k\in V(G)$ instead of $\exists S\subseteq V(G)$ such that $size(S)\leq k$, as the $d$-quasi-forest is hereditary.}

\begin{theorem}(Courcelle)
Given a graph $G$ and a formula $\varphi$ in  Monadic Second Logic describing a property of interest, and parameterizing by the combination of $tw(G)$  and the size of the formula $\varphi$,
it can be determined in time $f(tw(G), |\varphi|)n^{O(1)}$ whether $G$ has the property of interest.
%The key player is the fragment of second-order logic in the language of graphs.
\end{theorem}

\begin{theorem}
The $d$-quasi-forest deletion problem is FPT parameterized with $k$ and $d$.
\end{theorem}

Unfortunately, the algorithm implied by Courcelle's Theorem may be several layers exponential.
Aiming at fully exploiting the problem structure and design a faster algorithm, we solve the $d$-quasi-forest deletion problem by the iterative compression approach. By Guessing the intersection of $Z$ and $X$, the problem is transformed into the following disjoint version.

 \quad

  { \textbf{Disjoint $d$-quasi-forest deletion}} \nopagebreak

    \emph{Instance:} An undirected graph $G$, an integer $k$ and a vertex set $Z\subseteq V(G)$ with $|Z|\leq k+1$ such that $G-Z$ is a $d$-quasi-forest.

    \emph{Parameter:} $k$.

    \emph{Output:} Decide if there exists a set $X\in V(G)$ with $|X|\leq k, X\cap Z=\emptyset$ such that $G-X$ is a $d$-quasi-forest.

\quad
%Given graph $G$, integer $k$ and a set of vertices $Z\subseteq V(G)$ such that $|Z|\leq k+1$ and $G-Z$ is a 1-quasi-forest, the problem is to find a set $X\subseteq V(G)$ such that $X\cap Z=\emptyset$, $|X|\leq k$ and $G-X$ is a 1-quasi-forest.

Denote $F=G-Z$, then $F$ is a $d$-quasi-forest, that is, each connected component in $F$ admits feedback vertex set of size at most $d$. Note that according to the algorithm in \cite{DBLP:conf/iwpec/IwataK19}, we can check whether a given graph is a $d$-quasi-forest in time $O^{*}(3.460^d)$.

%we want to solve the problem via branching.

\textbf{Reduction Rule 1:}
If there is a vertex $u$ with degree at most one, then delete $u$ and return a new instance $(G-u,k)$.

\textbf{Reduction Rule 2:}
If there is a vertex $u$ with degree exactly two in $G$, then delete $u$ and add an edge between the neighbors of $u$.

After exhaustive applications of Reduction Rules 1-2, the resulting instance has minimum degree at least 3.

%\textbf{Reduction Rule 3:}
\textbf{Reduction Rule 3:} Observe that $G[Z]$ is a $d$-quasi-forest, otherwise it is a no-instance. If there is any vertex $u\in V(G)$ such that $G[u\cup V(Z)]$ is not a $d$-quasi-forest, then delete $u$ and decrease $k$ by one.

\begin{lemma}
Suppose $(G,Z,k)$ is a yes instance of Disjoint $d$-quasi-forest deletion, then for each vertex $u\in Z$, there are at most $k+d$ components in $G-Z$ that contains at least one cycle adjacent with $u$.
\end{lemma}
\begin{proof}
If there are more than $k+d$ components in $G-Z$, which are adjacent to $u$, contains a cycle, then for any vertex set $X\subseteq V(G-Z)$, such that $|X|\leq k$, the component containing $u$ in $G-X$ is not a $d$-quasi-forest. Thus $(G,Z,k)$ is a no instance. It follows that for each vertex $u\in Z$, the sum of the $fvs$ for the components in $G-Z$ that are adjacent to $u$ is at most $k+d$.
\end{proof}
By branching on minimum feedback vertex set of each such component in $G-Z$, i.e. either put the vertex into the solution or put it into $Z$, we may obtain a new instance in which every component in $G-Z$ is  a tree, moreover, the size of $Z$ is upper bounded by $k+1+(k+1)(k+d)=(k+1)(k+d+1)$.

%\textcolor{red}{Consider the measure $\mu(I)= k + cc(Z)-fvs(Z)+\omega'(Z)$ in which $cc(Z)$ is the number of components in $G[Z]$ and $\omega'(Z)$ is the maximum number of components in $G[Z-D]$ where $D$ is a minimum feedback vertex set of $G[Z]$.} Note that initially $|Z|\leq k+1$ and thus $\mu(I)\leq 3k+2$.

The following lemma provides a bound on the number of trees in $G-Z$ that has large neighborhood in $Z$.

\begin{lemma}\label{lemma:smallNeighbor}
Suppose $(G,Z,k)$ is a yes instance of Disjoint $d$-quasi-forest deletion, then there are at most  $2(k+1)+d$ components in $G-Z$ that has at least $d+2$  neighbors in $Z$.
\end{lemma}
\begin{proof}
Consider the measure $\mu(I)= cc(Z)+d-fvs(Z)+\omega'(Z)$, in which $cc(Z)$ is the number of components in $G[Z]$ and $\omega'(Z)$ is the maximum number of components in $G[Z-D]$ where $D$ is a minimum feedback vertex set of $G[Z]$.
Note that by putting a component $C$ that has at least $d+2$ neighbors in $Z$ into $Z$, either $cc(Z)$ decreases, or $fvs(Z)$ increases, or $\omega'(Z)$ decreases. Indeed, if $fvs(Z)$ does not increase, then $D$ contains no vertex in $C$, moreover, $D$ contains at most $d$ vertices, thus $C$ connects two components in $G[Z-D]$, and so $\omega'(Z)$ decreases. Since $cc(Z)+d-fvs(Z)+\omega'(Z)\leq 2(k+1)+d$, $G-Z$ of any yes instance contains at most $2(k+1)+d$ trees that has at least $d+2$ neighbors in $Z$.
\end{proof}
\begin{lemma}
Let $(G,Z,k)$ be a yes instance of disjoint $d$-quasi-forest deletion, then for each tree in $G-Z$ that has at least $d+2$ neighbors in $Z$, we may partition it into less than $2(2k+d+3)$ subtrees, each has at most $d+1$ neighbors in $Z$, and keep the number of vertices in $Z$ upper bounded.
\end{lemma}
\begin{proof}
%For any large tree with at least $d+2$ neighbors in $Z$, we can partition it into at most $2(2k+d+3)$ trees each has at most $d+1$ neighbors in $Z$.
Let $T$ be a tree in $G-Z$ that has more than $d+1$ neighbors in $Z$. Suppose on the contrary, we can only partition $T$ into at least $2(2k+d+3)$ maximal trees each has at most $d+1$ neighbors in $Z$, then there is a partition of $T$ into $2k+d+3$ smaller trees each has at least $d+2$ neighbors in $Z$(just combining two adjacent subtrees), which is not possible, according to Lemma \ref{lemma:smallNeighbor}. By branching on the boundaries of the at most $2(2k+d+3)$ trees, We reduce the instance into bounded number of new instances, in which each tree has at most $d+1$ in $Z$. Moreover, $Z$ has bounded size, as for each such tree, we put less than $2(2k+d+3)$ vertices into $Z$.
\end{proof}
Now we obtain an instance in which $G-Z$ consists of only trees each has at most $d+1$ neighbors in $Z$, where
$|Z|< (k+1)(k+d+1)+ 2(2k+d+3)(2(k+1)+d)=O(k^2d^2)$. To obtain an FPT algorithm, we further reduce the number of trees in $G-Z$.

\begin{definition}
Two trees $T_1, T_2$ in $G-Z$ have same \textbf{neighborhood type} in $Z$ if $N_{Z}(T_1)=N_{Z}(T_2)$ and for any vertex $u\in N_{Z}(T_1)$, $u$ has only one edge to $T_1$ if and only if $u$ has only one edge to $T_2$.
\end{definition}

\textbf{Reduction Rule 4:} For each neighborhood type $\sigma$, reduce the number of trees in $G-Z$ that have neighborhood type $\sigma$ in $Z$ to $k+d+2$.

\begin{lemma}\label{lemma4}
Reduction Rule 4 is safe.
\end{lemma}
\begin{proof}
Let $(G,Z,k)$ be an instance of disjoint $d$-quasi-forest deletion, such that there are $k+d+3$ trees with same neighborhood type $\sigma$ in $Z$, which are $T_1, T_2,\ldots,T_{k+d+3}$. By deleting $T_{k+d+3}$ we get a new instance $(G',Z,k)$. Let $N \subseteq Z$ be the common neighborhood of the trees in $G-Z$.

To prove the safety of Reduction Rule 4, we need to show that $(G,Z,k)$ is a yes instance if and only if $(G',Z,k)$ a yes instance.

On the one hand, it is easy to see that if $(G,Z,k)$ is a yes instance then $(G',Z,k)$ a yes instance as $G'$ is a subgraph of $G$.

On the other hand, suppose $(G',Z,k)$ is a yes instance. Then there is a set $X\subseteq G-Z$ such that $|X|\leq k$ and each component in $G-X$ admits a feedback vertex set of size at most $d$. Note that at least $d+2$ of $\{T_1, T_2,\ldots,T_{k+d+2} \}$ are disjoint from $X$, and every two vertices in $N$ will be connected by each of these trees. It follows that $G-X$ contains at least $d+2$ vertex disjoint paths between every two vertices in $N$. All vertices in $N$ are in the same connected component $C$ in $G-X$, thus all except one vertex in $N$ will have to be in the feedback vertex set of $C$.

 Let $D$ be a feedback vertex set of $C$ with size at most $d$.  Note that $T_1, T_2,\ldots,T_{k+d+3}$ are of same neighborhood type in $Z$, thus, for any vertex $u\in N$, if $u$  forms a cycle with $T_{k+d+3}$, then it forms a cycle with $T_{i}$, for any $i\in [k+d+2]$. Thus if $N-D$ is not empty, then $T_{k+d+3}$ does not form any cycle with the vertex in $N-D$. And so $D$ is still a feedback vertex set of $C$ when we put $T_{k+d+3}$ back. It follows that $(G',Z,k)$ is also a yes instance.

\end{proof}
For each $M\subseteq Z$, there are $2^{|M|}$ different neighborhood types with neighborhood $M$, depending on whether the number of edges each vertex in $M$ has to the trees is just one.

According to Lemma \ref{lemma4}, after exhaustive applications of Reduction Rule 4, there are at most $k+d+2$ trees of each neighborhood type. %Here we just require that $|M|\leq d+1$.

When each tree in $G-Z$ has at most $d+1$ neighbors in $Z$, there are at most $\sum_{1\leq i\leq d+1}\binom{|Z|}{i}2^i$ different neighborhood types, and for each neighborhood type, Reduction Rule 4 can be applied in polynomial time.

%\textcolor{green}{Regard a large tree as many small trees.} We want to find a subset $S$ of a large tree $T$, such that each tree in $T-S$ is small.

%To reduce redundancy, by replacing the large tree with a smaller one.

%\textbf{Reduction Rule 5:} Let $T$ be a tree in $G-Z$ which has at least $|d+1| f(k,d) $ edges to $Z$, then we may replace $T$ with a smaller tree $T'$ and get a new instance $(G',Z)$.

%We use Gallai Theorem to prove degree bound.
\begin{theorem}(Gallai)
Given a simple graph $G$, a set $R \subseteq V(G)$, and an integer $s$, one can in polynomial time either
\begin{enumerate}
\item find a family of $s +1$ pairwise vertex-disjoint $R$-paths, or
\item conclude that no such family exists and, moreover, find a set $B$ of at most $2s$ vertices, such that in $G \backslash B$ no connected component contains more than one vertex of $R$.
\end{enumerate}
\end{theorem}

%\textcolor{red}{Try to prove that for any subset $N_0$ of $N$, there is no need to have more than $k+d+2$ vertices each has $N_0$ as its neighborhood in $Z$. If this is true, maybe we can replace $T$ with a tree with at most $2^{|N|}(k+d+2)$ vertices.}

%\begin{lemma}\label{lemma5}
%Reduction Rule 5 is safe.
%\end{lemma}
%\begin{proof}
%As the tree has $|d+1| f(k,d) $ edges to $Z$, if the tree has small neighborhood $N$ in $Z$, then at least one vertex in $Z$ has at least $f(k,d)$  edges to the tree. \textcolor{red}{We require that $T'$ is the same with $T$ for those vertices in $N$ with small degree.}
%\end{proof}

\begin{definition}
%We called a vertex $u\in Z$ \textbf{\textit{forced}} if $u$ must be in every feedback vertex set of size at most $d$ of the component containing $u$ in $G-X$ for any solution $X$.

We called a vertex $u\in Z$ \textbf{\textit{forced}} if for every solution $X$, every feedback vertex set of size at most $d$ of the component containing $u$ in $G-X$ must contain $u$.
\end{definition}

%\textcolor{red}{introduce some notation to simplify the description.}
Note that a vertex $u\in Z$ is a forced vertex if $G-Z$ contains $k+d+1$ vertex disjoint paths between neighbors of $u$.
Thus we may set $s=k+d$, and then for each vertex  $u\in N_Z(T)$, let $R=N_{T}(u)$, check whether $u$ is a forced vertex via Gallai's Theorem.
 %There are bounded number of trees in $G-Z$ when Reduction Rule 4 is not applicable. For each large tree, we do branches to get smaller trees. Hopefully, we don't need to do it many times, such that the size of $Z$ is always bounded.

%\textcolor{red}{bounded number of trees with small neighborhood, but the trees does not have bounded size $\rightarrow$  bounded number of small trees. }

We already show that the number of trees in $G-Z$ is upper bounded. We may now guess which trees are intersecting the solution, note that there are at most $k$ such trees. For each guess, we check whether it is compatible, i.e. does putting all the trees (that are guessed to be disjoint with the solution) into $Z$ violate the requirement of $d$-quasi-forest.

We guess how the components in $Z$ are going to be connected in $G-X$. For each compatible guess, we need to delete vertices from the trees, such that, the components are separated accordingly. That is, we need to delete vertices from the trees, such that, some components are not going to be connected.

There are bounded number of vertices in $Z$ that are neighbors of $G-Z$. We can compute the size of minimum feedback vertex set for each component in $G[Z]$. %We can afford to guess whether two components are going to be connected by the trees. We can afford to guess how many vertices in each tree are in the solution.
We may also guess for each vertex in the neighborhood of $G-Z$ in $Z$, whether it is in the feedback vertex set or not.

We label each component with the size of its minimum feedback vertex set.

%\textcolor{red}{One thing to be careful is that, the number of trees will increase after we do the decomposition, so we need to apply reduction rule 4 once again.}

%For each compatible, we need to decide whether we can delete at most $d_i$ vertices from $T_i$, such that it does not connect two components which we guess to be separated.

%Now we need solve the instance in which each tree has small neighborhood in $Z$, $Z$ has bounded size and the number of trees in $G-Z$ is upper bounded. We may try the following idea of guessing intersection of each tree with the solution. We may obtain bounded number of instances in which there are at most $k$ trees in $G-Z$, each tree has small neighborhood in $Z$.

%The trees has small neighborhood in $Z$, thus there are two types of effects of each tree, one is connecting components in $Z$, and the other is forming cycles with its neighbors in $Z$.

%We may afford to guess which vertices are in same component for vertices in $N$. Just consider those guesses that are compatible with the current $G[Z]$. First computing the minimum feedback vertex set size of each component in $G[Z]$. For each vertex in $N$, we may also guess whether it is in the feedback vertex set.

\begin{lemma}
The problem reduces to bounded number of instances in which $G-Z$ consists of at most $k$ trees, each has at most $d+1$ neighbors in $Z$, moreover, each tree contains at most $k\binom{p}{2}+k|N_2|+(k+1)(k\binom{p}{2}+k|N_2|)$ vertices, where $p\leq k+1$.
\end{lemma}
\begin{proof}

 Since $G-Z$ contains at most $k$ trees, we know that $N=N_{Z}(G-Z)$ contains at most $k(d+1)$ vertices. For each vertex $u\in N$, we decide whether it is a forced vertex via Gallai's Theorem. Thus we obtain a partition of $N=N_1\cup N_2$, where $N_1$ contains all the forced vertices. According to Gallai's Theorem, for each vertex $u\in N_2$, we can find a set $B_u$ with at most $2(k+d)$ vertices, such that each component in $G-Z-B_u$ contains at most one neighbor of $u$. Thus, by branching on whether each vertex in $B_u$ is in the solution, we may partition each tree $T_i$ in $G-Z$ into subtrees, which contains at most one neighbor of each vertex in $B_u\cup N_2$.

We apply reduction rule 4  again to reduce the number of trees. And then we further reduce the instance by guessing how the trees are going to intersecting  with the solution. Easy to know that, there are at most $k$ trees in the reduced instance, moreover, each tree in $G-Z$ contains at most one neighbor of each vertex in $N_2\cup (\cup_{u\in N_2}B_u)$.

 Now we show how to bound the size of each tree.  We guess how the component are going to be separated in $G-X$ where $X$ is a minimum $d$-quasi-forest deletion set of $(G,Z,k)$. There are bounded number of such guesses. And then, we check whether the guessing can be realized. To simplify, we may regard all vertices in the same component from $N_1$ as one vertex. For two vertices $u,v\in N_1$ that are guessed to be in different components, there are at most $k$ vertices in the trees that are adjacent to both $u$ and $v$, since otherwise, deleting $X$ cannot separate $u$ and $v$.  Thus the total number of vertices that have neighbors in two guessed components must be upper bounded. Suppose we guessed there to be $p$ components, then at most $\binom{p}{2}$ pairs of vertices in the components. It follows that there are at most $k\binom{p}{2}$ vertices in the trees that have neighbors in different guessed components.

 And we also know that in each tree, the number of vertices that are adjacent to $N_2\cup (\cup_{u\in N_2}B_u)$ is upper bounded, since each tree contains at most one neighbor of each vertex in $N_2\cup (\cup_{u\in N_2}B_u)$. Thus we only need to bound the number of vertices that are not adjacent to $N_2\cup (\cup_{u\in N_2}B_u)$ and only have neighbors in just one component. We bound this by arguing that there is no need to keep too many such vertices in the same tree, since such vertices only connecting vertices that are guessed to be in the same components, moreover, their neighbors in $Z$ are forced to be in the feedback vertex sets, thus the number of such vertices does not affect the solution. In each tree, for each path connecting neighbors  of $N_2\cup (\cup_{u\in N_2}B_u)$ and vertices with neighbors in more than one components, we just keep at one vertex that are adjacent to one component in $Z$(note that there are at most $k+1$ components).

 Thus we obtain a bound on the size of each tree in $G-Z$, moreover, there are at most $k$ tree. Each tree in the reduced instance will have at most $k\binom{p}{2}+k|N_2|+(k+1)(k\binom{p}{2}+k|N_2|)$ vertices. And so we can solve the reduced instance in fpt time by branching on the vertices in $G-Z$, which has bounded number of vertices.

\end{proof}

\begin{theorem}
The $d$-quasi-forest deletion problem can be solved in time \\$O^*({{c_1}^{(k+d)}}^{c_2(k+d)})$.
\end{theorem}
\begin{proof}
We solve the Disjoint $d$-quasi-forest deletion problem as a subroutine. It takes at most $O^*((2^{k+d})^{k+1})$ steps to make $G-Z$ acyclic. And it takes at most $O^*({2^{2(2k+d+3)}}^{2k+2+d}) $ steps to apply lemma 6.
It take $\sum_{1\leq i\leq d+1}\binom{|Z|}{i}2^i$  steps to exhaustively apply Reduction Rule 4.
It takes  at most $2^{2k+d+2}$  guesses to get an instance in which there are at most $k$ trees in $G-Z$.
For each instance, in which there are at most $k$ trees each has at most $d+1$ neighbors in $Z$, we can reduce it into bounded number of instances, where each tree has bounded size. And finally, we just need to solve instances in which $G-Z$ contains at most $k$ trees, each has at most $k\binom{p}{2}+k|N_2|+(k+1)(k\binom{p}{2}+k|N_2|)=O(k^4d)$ vertices. Thus it takes at most $O^*(2^{k^5d})$ steps to solve each such instance via branching. According to the above analysis, the running time of our algorithm for Disjoint $d$-quasi-forest deletion is at most $(2^{k+d})^{k+1} {2^{2(2k+d+3)}}^{2k+2+d} \sum_{1\leq i\leq d+1}\binom{|Z|}{i}2^i 2^{2k+d+2} 2^{k^5d} = O^*({{c_1}^{(k+d)}}^{c_2(k+d)})$. To solve $d$-quasi-forest deletion, it suffices to solve $2^{|Z|}\leq 2^{k+1}$ copies of Disjoint $d$-quasi-forest deletion, thus, we may solve $d$-quasi-forest deletion in time  $O^*({{c_1}^{(k+d)}}^{c_2(k+d)})$.
\end{proof}

\section{Conclusion}
In this paper, we provide FPT results for two generalized versions of feedback vertex set problem.
It would be interesting to know whether the problem of $d$-quasi-forest deletion admits a polynomial kernel.
\section{Acknowledgement}
The research of Bin Sheng is supported by National Natural Science Foundation of China (No. 61802178).
\bibliographystyle{plain}
\cleardoublepage
\bibliography{references}
\end{document}